\RequirePackage{fix-cm}
\documentclass[smallcondensed]{svjour3}       
\smartqed  
\usepackage{graphicx}
\usepackage{mathptmx}      
\usepackage{mathrsfs}
\usepackage{amsbsy}
\usepackage{amsfonts}
\usepackage{enumerate}
\usepackage{bbm}
\usepackage{physics}
\usepackage{hyperref}
\usepackage{xcolor}
\usepackage{cite}

\usepackage[11pt]{extsizes}
\usepackage[margin=2.5cm]{geometry}
\def\makeheadbox{\relax}
\date{}

\newcommand{\Del}{\nabla}

\renewcommand{\H}{\mathcal{H}}
\renewcommand{\tilde}{\widetilde}
\newcommand{\herm}{\operatorname{Herm}}
\newcommand{\pos}{\operatorname{Pos}}
\newcommand{\psd}{\operatorname{PSD}}
\renewcommand{\d}{\dd\hspace{-0.04em}}

\begin{document}

\title{Approximate Petz recovery from the geometry of density operators\thanks{The authors are supported by AFOSR award FA9550-19-1-0369, CIFAR, DOE award DE-SC0019380 and the Simons Foundation.}
}

\author{Sam Cree         \and
        Jonathan Sorce
}

\institute{Sam Cree and Jonathan Sorce \at
              Stanford Institute for Theoretical Physics, 382 Via Pueblo, Stanford, CA 94305
           \and
           Sam Cree \at
              \email{scree@stanford.edu}
           \and
           Jonathan Sorce \at
				\email{jsorce@stanford.edu}
}


\maketitle

\begin{abstract}
	We derive a new bound on the effectiveness of the Petz map as a universal recovery channel in approximate quantum error correction using the second sandwiched R\'{e}nyi relative entropy $\tilde{D}_{2}$. For large Hilbert spaces, our bound implies that the Petz map performs quantum error correction with order-$\epsilon$ accuracy whenever the data processing inequality for $\tilde{D}_{2}$ is saturated up to terms of order $\epsilon^2$ times the inverse Hilbert space dimension. Conceptually, our result is obtained by extending \cite{cree2020geometric}, in which we studied exact saturation of the data processing inequality using differential geometry, to the case of approximate saturation. Important roles are played by (i) the fact that the exponential of the second sandwiched R\'{e}nyi relative entropy is quadratic in its first argument, and (ii) the observation that the second sandwiched R\'{e}nyi relative entropy satisfies the data processing inequality even when its first argument is a non-positive Hermitian operator.
\end{abstract}

\section{Introduction}
\label{sec:intro}

In \cite{petz1986sufficient, petz1988sufficiency}, Petz proved a very general property of von Neumann algebras whose analogue for finite-dimensional quantum systems can be stated as follows. Let $\H$ and $\H'$ be finite-dimensional Hilbert spaces of dimension $d$ and $d'$. A \emph{density operator} on $\H$ is a unit-trace, positive semidefinite operator on $\H$. Let $\sigma$ be a strictly positive density operator on $\H$, and let $\rho$ be any density operator on $\H$. Define the relative entropy of $\rho$ with respect to $\sigma$ by
\begin{equation}
	D(\rho \| \sigma) \equiv \tr(\rho \log \rho) - \tr(\rho \log \sigma).
\end{equation}
Let $\Lambda$ be a quantum channel from $\H$ to $\H'$, i.e., a completely positive, trace-preserving (CPTP) linear map from the operator space $\mathcal{L}(\H)$ to the operator space $\mathcal{L}(\H').$ Petz's theorem states that if the relative entropy is unchanged under application of the quantum channel $\Lambda$, i.e.,
\begin{equation}
	D(\rho \| \sigma)
		= D(\Lambda(\rho) \| \Lambda(\sigma)),
\end{equation}
then there exists a quantum channel $\mathcal{R}$ from $\H'$ to $\H$ that recovers the states $\rho$ and $\sigma$, i.e.,
\begin{equation} \label{eq:Petz-recovery}
	(\mathcal{R} \circ \Lambda)(\rho) = \rho
	\qquad \text{and}
	\qquad 
	(\mathcal{R} \circ \Lambda)(\sigma) = \sigma.
\end{equation}
The channel $\mathcal{R}$, which depends on the state $\sigma$ and the channel $\Lambda$, can be written explicitly as the \emph{Petz map}
\begin{equation} \label{eq:petz-map}
	\mathcal{R}_{\sigma, \Lambda}(M)
		= \sigma^{1/2} \Lambda^* \left[ \Lambda(\sigma)^{-1/2} \cdot M \cdot \Lambda(\sigma)^{-1/2} \right] \sigma^{1/2}.
\end{equation}
The map $\Lambda^*$ appearing in this expression is the adjoint of the quantum channel $\Lambda$, defined with respect to the Hilbert-Schmidt inner product on operators $\langle A, B \rangle = \tr(A^{\dagger} B).$

In some situations, Petz's theorem can be used to prove that a Hilbert space $\H$, often thought of as a ``logical'' subspace of a larger ``physical'' Hilbert space $\H'$, forms a quantum error correcting code for a channel $\Lambda$. If one can prove that for all density operators $\rho$ on $\H$, the relative entropy of $\rho$ with respect to the maximally mixed state $I/d$ is invariant under $\Lambda$, then equation \eqref{eq:Petz-recovery} implies that the Petz map $\mathcal{R}_{I/d, \Lambda}$ is a CPTP inverse for $\Lambda$ on all of $\H$. The existence of a CPTP map inverting $\Lambda$ on all of $\H$ is exactly the condition for $\H$ to be a quantum error correcting code for $\Lambda.$

Recent years have seen renewed interest in the theory of \emph{approximate} quantum error correction. In practical applications, it may be too much to ask that a Hilbert space $\H$ admits a perfect CPTP inverse for $\Lambda$ --- it is often good enough to have a quantum channel $\mathcal{R}$ for which $\mathcal{R}\circ \Lambda$ is approximately the identity on $\mathcal{L}(\H)$ in some appropriate sense. In \cite{junge2018universal}, it was established that a generalization of the Petz map called the \emph{twirled} or \emph{rotated} Petz map $\mathcal{T}_{\sigma, \Lambda}$ and defined by
\begin{equation}
	\mathcal{T}_{\sigma, \Lambda}(M)
		= \frac{\pi}{2} \int \d t\, \frac{1}{1 + \cosh{\pi t}} \sigma^{- i t/2} \mathcal{R}_{\sigma, \Lambda}(\Lambda(\sigma)^{i t/2} M \Lambda(\sigma)^{- i t/2}) \sigma^{i t/2},
\end{equation}
satisfies the equations
\begin{equation}
	(\mathcal{T}_{\sigma, \Lambda}\circ \Lambda)(\sigma)
		= \sigma
\end{equation}
and
\begin{equation} \label{eq:twirled-Petz-inequality}
	\lVert \rho - (\mathcal{T}_{\sigma, \Lambda} \circ \Lambda)(\rho) \rVert_1
		\leq 2 \sqrt{1 - \exp(- \Delta D)}.
\end{equation}
In these expressions, the symbol $\Delta D$ denotes the relative entropy difference 
\begin{equation}
	\Delta D
		\equiv D(\rho \| \sigma) - D(\Lambda(\rho) \| \Lambda(\sigma)),
\end{equation}
and the symbol $\lVert \cdot \rVert_1$ denotes the one-norm
\begin{equation}
	\lVert A \rVert_1 \equiv \tr(\sqrt{A^{\dagger} A}).
\end{equation}
If the relative entropy of $\rho$ with respect to $\sigma$ is \emph{approximately} constant under the application of $\Lambda$, then a Taylor series expansion gives 
\begin{equation}
	\lVert \rho - (\mathcal{T}_{\sigma, \Lambda} \circ \Lambda)(\rho) \rVert_1
		\leq 2 \sqrt{\Delta D} + O((\Delta D)^{3/2}).
\end{equation}
So if all the relative entropies in a Hilbert space $\H$ with respect to some fixed reference state $\sigma$ --- usually chosen to be the maximally mixed state $I/d$ --- are approximately preserved under $\Lambda$, then the corresponding rotated Petz map is an approximate inverse for $\Lambda$ for all states on $\H$.

It was then shown in \cite{chen2020entanglement}, building on earlier work in \cite{barnum2002reversing}, that the Petz map for the maximally mixed state is not too much worse than any twirled Petz map for sufficiently small Hilbert spaces. Specifically, the authors showed that if there exists \emph{any} recovery map $\mathcal{R}$ satisfying
\begin{equation}
	\lVert \rho - (\mathcal{R} \circ \Lambda)(\rho) \rVert_1
		< \delta
\end{equation}
for the density operator $\rho$, then the Petz map $\mathcal{R}_{I/d, \Lambda}$ satisfies
\begin{equation}
	\lVert \rho - (\mathcal{R}_{I/d, \Lambda} \circ \Lambda)(\rho) \rVert_1
		< d \sqrt{8 \delta}.
\end{equation}
Comparing this with the twirled Petz map inequality \eqref{eq:twirled-Petz-inequality}, we learn that the Petz recovery map with maximally mixed reference state satisfies
\begin{equation} \label{eq:petz-rel-ent-bound}
	\lVert \rho - (\mathcal{R}_{I/d, \Lambda} \circ \Lambda)(\rho) \rVert_1
	< 4 d \left[1 - \exp(- \Delta D_{I/d}) \right]^{1/4},
\end{equation}
with 
\begin{equation}
	\Delta D_{I/d} = D(\rho \| I/d) - D(\Lambda(\rho) \| \Lambda(I/d)).
\end{equation}
This inequality contains an interesting balance between the Hilbert space dimension $d$ and the quantity $\Delta D_{I/d}$. For the Petz map with maximally mixed reference state to be a good recovery map, $d \cdot (\Delta D_{I/d})^{1/4}$ must be small, which requires the approximate preservation of relative entropy to scale as $d^{-4}$ for large Hilbert spaces.

The purpose of this paper is to report a new bound on the quality of the Petz map for approximate quantum error correcting codes, in terms of a generalization of the relative entropy called the \emph{second sandwiched R\'{e}nyi relative entropy}. This quantity, which we will denote by the symbol $\tilde{D}_{2}$, is defined by
\begin{equation}
	\tilde{D}_{2}(\rho \| \sigma)
		\equiv \log\tr[(\sigma^{-1/4} \rho \sigma^{-1/4})^{2}].
\end{equation}
It is a member of a continuous family of sandwiched R\'{e}nyi relative entropies \cite{muller2013quantum,wilde2014strong} that has been studied extensively in the quantum information literature (see, e.g., \cite{frank2013monotonicity, beigi2013sandwiched, leditzky2017data, wang2020alpha}). $\tilde{D}_{2}$ satisfies many of the important properties of the ordinary relative entropy, in particular the monotonicity inequality (or ``data processing inequality'')
\begin{equation} \label{eq:monotonicity}
	\tilde{D}_{2}(\rho \| \sigma) \geq \tilde{D}_{2}(\Lambda(\rho) \| \Lambda(\sigma))
\end{equation}
and positivity $\tilde{D}_{2}(\rho \| \sigma) \geq 0.$

The most general bound we prove in this paper, for arbitrary strictly positive $\sigma$ and arbitrary \emph{Hermitian} $\rho$,\footnote{It may be surprising that this bound holds for all Hermitian $\rho$, and does not require $\rho$ to be a density operator. We discuss this feature of the calculation in detail in section \ref{sec:conceptual-section}.} is 
\begin{equation} \label{eq:prelim-bound}
	\lVert \rho - (\mathcal{R}_{\sigma, \Lambda} \circ \Lambda)(\rho) \rVert_1
		\leq \lVert \sigma \rVert_2 \times \sqrt{\lVert \sigma^{-1} \rVert_\infty}
				\times \sqrt{e^{\tilde{D_2}(\rho \| \sigma)} - e^{\tilde{D_2}(\Lambda(\rho) \| \Lambda(\sigma))}},
\end{equation}
where $\lVert \cdot \rVert_p$ denotes the Schatten $p$-norm $\lVert A \rVert_p = \tr(|A|^p)^{1/p}.$ For $\sigma$ the maximally mixed state $I/d,$ this simplifies to
\begin{equation} \label{eq:Q-bound}
	\lVert \rho - (\mathcal{R}_{I/d, \Lambda} \circ \Lambda)(\rho) \rVert_1
		\leq \sqrt{e^{\tilde{D_2}(\rho \| I/d)} - e^{\tilde{D_2}(\Lambda(\rho) \| \Lambda(I/d))}}.
\end{equation}
Using a technique described in section \ref{sec:approximate-section}, we show that this bound implies our main result,
\begin{equation} \label{eq:main-bound}
	\lVert \rho - (\mathcal{R}_{I/d, \Lambda} \circ \Lambda)(\rho) \rVert_1
		\leq d^{1/2} \left[ 1 - \exp(- \Delta \tilde{D}_{2, I/d}) \right]^{1/2},
\end{equation}
with
\begin{equation}
	\Delta \tilde{D}_{2, I/d} = \tilde{D}_2(\rho \| I/d) - \tilde{D}_2(\Lambda(\rho) \| \Lambda(I/d)).
\end{equation}
Naively, the scaling in this bound appears to be better than the scaling in \eqref{eq:petz-rel-ent-bound}. For that bound to be good at large $d$, it was necessary for the change in relative entropy under $\Lambda$ to scale as $d^{-4}.$ In \eqref{eq:main-bound}, it is only necessary for the change in the second sandwiched R\'{e}nyi relative entropy to scale as $d^{-1}.$ It is hard to say for sure whether \eqref{eq:main-bound} is truly better than \eqref{eq:petz-rel-ent-bound}, though, because we are unaware of any general bounds relating the magnitudes of $\Delta \tilde{D}_{2, I/d}$ and $\Delta D_{I/d}$. One could imagine, for example, a situation where both quantities are small compared to one, but $\Delta \tilde{D}_{2, I/d}/\sqrt{\Delta D_{I/d}}$ is of order $d$; in this case, the two inequalities are of the same order in the large-$d$ limit. However, upon generating $1,000$ random states $\rho$ and channels $\Lambda$ in each dimension $d=2$ through $d=15$ according to the specifications of our appendix, and computing the maximum of $\Delta \tilde{D}_{2, I/d}/\sqrt{\Delta D_{I/d}}$ within our sample set, we found values that displayed no linear scaling --- e.g. $\sim 1.14$ for $d=5,$ $\sim 1.13$ for $d=10,$ and $\sim 1.12$ for $d=15.$ So even if there do exist special classes of states and channels for which the ratio scales linearly in $d$, such states do not appear to be generic. Furthermore, even if there are cases with $\Delta \tilde{D}_{2, I/d} \sim d \sqrt{\Delta D_{I/d}}$, this isn't so important for practical applications --- so long as one is working in a situation where $\Delta \tilde{D}_{2,I/d}$ is provably small, equation \eqref{eq:main-bound} has practical utility, and it is often easier to compute $\tilde{D}_{2}$ than $D.$

The essential machinery of the proof of inequality \eqref{eq:main-bound} originated in our previous paper \cite{cree2020geometric}. One begins by defining the trace functional $\tilde{Q}_{2}$ by
\begin{equation} \label{eq:Q-2}
	\tilde{Q}_{2}(\rho \| \sigma) = e^{\tilde{D}_{2}(\rho \| \sigma)} = \tr\left[ (\sigma^{-1/4} \rho \sigma^{-1/4})^2 \right],
\end{equation}
and then thinks of the function\footnote{While the operator $\Lambda(\sigma)^{-1/2}$ is not necessarily well defined, as $\Lambda(\sigma)$ need not be invertible even though $\sigma$ is, the fact that the support of $\rho$ is contained within the support of $\sigma$ means that the support of $\Lambda(\rho)$ is contained in the support of $\Lambda(\sigma)$. (See e.g. Lemma 3.1 of \cite{hiai_quantum_2011}.) This implies that $\tilde{Q}_{2}(\Lambda(\rho) \| \Lambda(\sigma))$ makes sense so long as we define $\Lambda(\sigma)^{-1/2}$ to be zero outside the support of $\Lambda(\sigma)$.}
\begin{equation} \label{eq:f-function}
	f_{\sigma, \Lambda}(\rho)
		= \tilde{Q}_{2}(\rho \| \sigma) - \tilde{Q}_{2}(\Lambda(\rho) \| \Lambda(\sigma))
\end{equation}
as a differentiable function on the Riemannian manifold of density operators. (We will also occasionally refer to this quantity as $\Delta \tilde{Q}_{2}$ or $\Delta \tilde{Q}_{2, \sigma}$.) While the function itself is scalar-valued, its gradient is a tangent vector on the manifold of density operators --- i.e., its gradient is an operator. Inequality \eqref{eq:monotonicity} together with monotonicity of the exponential function implies that $f_{\sigma, \Lambda}$ is nonnegative, so if it equals zero at a point of its domain, then that point is a minimum of the function and so the gradient must vanish. While this same calculation can be performed for any differentiable function of density operators satisfying a monotonicity inequality like \eqref{eq:monotonicity}, $\tilde{Q}_{2}$ has two special properties that make it especially well suited to this sort of analysis:
\begin{enumerate}[(i)]
	\item  The function $f_{\sigma, \Lambda}(\rho)$ extends smoothly to the space of all Hermitian operators on $\H$, and is still nonnegative on this domain. This allows the analysis to be applied without caveats even on the boundary of the space of density operators, i.e., even when $\rho$ has one or more vanishing eigenvalues.
	\item The gradient of $f_{\sigma, \Lambda}$ can be expressed in terms of the Petz recovery operator $\rho - (\mathcal{R}_{\sigma, \Lambda} \circ \Lambda)(\rho)$, so that vanishing of $f_{\sigma, \Lambda}$ implies $(\mathcal{R}_{\sigma, \Lambda} \circ \Lambda)(\rho) = \rho.$
\end{enumerate}
We discuss both these properties in detail in section \ref{sec:conceptual-section}. In section \ref{sec:approximate-section}, we will show that $\tilde{Q}_{2}$ has the additional property that when $f_{\sigma, \Lambda}$ is \emph{small but nonzero}, its gradient is small as an operator in an appropriate sense. By making that statement precise, we will be able to derive inequality \eqref{eq:main-bound}.

Before proceeding to the plan of the paper, we pause to comment on the relationship between our inequality \eqref{eq:prelim-bound} and some existing inequalities governing approximate Petz recovery for sandwiched R\'{e}nyi relative entropies, derived in theorem 4.20 of \cite{gao2021recoverability}. The bound derived there for the second sandwiched R\'{e}nyi relative entropy is
\begin{equation} \label{eq:Marks-bound}
	\lVert \rho - (\mathcal{R}_{\sigma, \Lambda} \circ \Lambda)(\rho) \rVert_1
		\leq \frac{2}{\pi} (\Delta \tilde{Q}_{2})^{1/4 - \epsilon}
							\left( 4 (\lVert \rho^{1/2} \sigma^{-1} \rho^{1/2} \rVert_{\infty})^{1/2} + \left( \frac{\pi}{e \epsilon} \right)^{1/2} + 4 \right),
\end{equation}
where this inequality holds for all $\epsilon \in (0, 1/4).$ This bound is generally weaker than our bound \eqref{eq:prelim-bound}. For example, with $\sigma = I/d,$ $d=10$, $\Delta \tilde{Q}_{2} = 1/100$ and $\rho$ pure, the right-hand side of \eqref{eq:Marks-bound} can be numerically minimized over the allowed range of $\epsilon$ with a minimum value of $\sim\hspace{-0.1cm}5.28$. By contrast, the right-hand side of our inequality \eqref{eq:prelim-bound} in this case is $0.1.$

The plan of the paper is as follows.

In section \ref{sec:conceptual-section}, we review the essential results of \cite{cree2020geometric} as they apply to the second sandwiched R\'{e}nyi relative entropy, emphasizing what makes $\tilde{Q}_{2}$ special as compared to other quantum distance measures.\footnote{We note that by ``distance measure,'' we do not mean a metric in the sense of a positive semidefinite, symmetric function satisfying the triangle inequality; we mean a positive semidefinite function of two density operators that is monotonically decreasing under the application of the same quantum channel to both arguments.} We explain how, in our geometric framework, exact saturation of monotonicity for $\tilde{Q}_{2}$ implies perfect recovery for the Petz map.

In section \ref{sec:approximate-section}, we show that approximate saturation of monotonicity for $\tilde{Q}_2$ implies an inequality on the one-norm quality of the Petz map, namely inequality \eqref{eq:prelim-bound}. By computing this inequality for the case where $\sigma$ is the maximally mixed state, we prove inequality \eqref{eq:main-bound}.

In an appendix, we provide a numerical test of the tightness of inequalities \eqref{eq:Q-bound} and \eqref{eq:main-bound} in $d=2$ and $d=3$ by generating random states $\rho$ and random channels $\Lambda$, then computing and comparing both sides of the inequalities on the resulting space of samples.

\section{Petz recovery and the second sandwiched R\'{e}nyi relative entropy}
\label{sec:conceptual-section}

For a finite-dimensional Hilbert space $\H$, the space of Hermitian operators $\herm(\H)$ is a real vector space with the natural Hilbert-Schmidt inner product $\langle A, B \rangle = \tr(A^{\dagger} B).$ The subset of \emph{positive} operators $\pos(\H)$ is not a vector space, but it retains the manifold structure of $\herm(\H)$ with Riemannian metric given by the inner product. The tangent space at each point is isomorphic to $\herm(\H)$, since $\rho + \epsilon M$ for positive $\rho$ remains positive at small $\epsilon$ if and only if $M$ is Hermitian.

For a differentiable map $F$ between manifolds $\mathcal{M}$ and $\mathcal{N}$, the \emph{derivative} of $F$ at a point $p$ is defined as a map on the tangent spaces $\d F|_{p} : T_p \mathcal{M} \rightarrow T_{F(p)} \mathcal{N}$ that is compatible with the local behavior of $F$. For maps between manifolds with a local vector space structure, the definition is quite simple to state, and is analogous to the definition of the derivative in single-variable calculus:
\begin{equation} \label{eq:vector-space-derivative}
	\d F|_{p}(M) = \lim_{\epsilon \rightarrow 0} \frac{F(p + \epsilon M) - F(p)}{\epsilon}.
\end{equation}
If $F$ is a map from $\pos(\H)$ to $\mathbb{R}$, and $\rho$ is an operator in $\pos(\H)$, then $\d F|_{\rho}$ is a map from $\herm(\H)$ to the real numbers, i.e., an element of the dual space $\herm(\H)^*.$ Non-degeneracy of the Hilbert-Schmidt inner product implies the existence of a unique Hermitian operator $\Del F|_{\rho}$ satisfying
\begin{equation} \label{eq:gradient-definition}
	\langle \Del F|_{\rho}, M \rangle
		= \d F|_{\rho}(M)
\end{equation}
for all Hermitian operators $M$. The operator $\Del F|_{\rho}$ is called the \emph{gradient} of $F$ at the point $\rho$.

Now, let $\sigma$ be a strictly positive density operator on $\H$, let $\tilde{Q}_{2}$ be defined as in equation \eqref{eq:Q-2}, and let $f_{\sigma, \Lambda}$ be defined as in equation \eqref{eq:f-function}. It was shown in section 3 of \cite{cree2020geometric} that the gradient of $f_{\sigma, \Lambda}$ is given by
\begin{equation}
	\Del f_{\sigma, \Lambda}|_{\rho}
		= 2 \sigma^{-1/2} \rho \sigma^{-1/2}
			- 2 \Lambda^*[\Lambda(\sigma)^{-1/2} \Lambda(\rho) \Lambda(\sigma)^{-1/2}].
\end{equation}
In terms of the Petz map \eqref{eq:petz-map}, this is
\begin{equation} \label{eq:Q2-gradient}
	\Del f_{\sigma, \Lambda}|_{\rho}
		= 2 \sigma^{-1/2} \left[\rho - (\mathcal{R}_{\sigma, \Lambda} \circ \Lambda)(\rho)) \right] \sigma^{-1/2}.
\end{equation}
As explained in the introduction, nonnegativity of $f_{\sigma, \Lambda}$ on $\pos(\H)$ implies that if $f_{\sigma, \Lambda}(\rho)$ vanishes, then $\Del f_{\sigma, \Lambda}|_{\rho}$ vanishes as well. In other words, we have the implication
\begin{equation} \label{eq:petz-implication}
	\tilde{Q}_{2}(\rho \| \sigma) = \tilde{Q}_{2}(\Lambda(\rho) \| \Lambda(\sigma))
		\Rightarrow \rho - (\mathcal{R}_{\sigma, \Lambda} \circ \Lambda)(\rho) = 0,
\end{equation}
so saturation of monotonicity for $\tilde{D}_{2}$, which implies saturation of monotonicity for $\tilde{Q}_{2}$, implies that the $\sigma$-Petz map inverts $\Lambda$ on the state $\rho.$

There is one caveat to this discussion that will become important in the next section, which is that while we have assumed that $\rho$ lies in $\pos(\H)$, there are perfectly good quantum states described by positive \emph{semidefinite} operators that are not strictly positive. These operators lie on the boundary of the larger space $\psd(\H)$ containing $\pos(\H).$ This case is subtle because, as is the case in ordinary single-variable calculus, global minima of a function do not need to have vanishing derivatives at points on the boundary of the function's domain. The vanishing of the gradient is still assured, however, provided that we first project the gradient, considered as a vector on the tangent space at $\rho$, into the subspace of tangent directions that lie along the boundary of $\psd(\H)$. The details of this projection were discussed in section 4 of \cite{cree2020geometric}. 

In the special case of $\tilde{Q}_{2}$, however, this subtlety does not arise. This is because $\tilde{Q}_{2}$ is monotonic \emph{for all Hermitian operators $\rho$}, so long as $\sigma$ is a strictly positive density operator. The more general family of sandwiched R\'{e}nyi relative entropies are defined by
\begin{equation}
	\tilde{D}_{\alpha}(\rho \| \sigma) = \frac{1}{\alpha - 1} \log \tilde{Q}_{\alpha}(\rho \| \sigma) 
\end{equation}
with
\begin{equation}
	\tilde{Q}_{\alpha}(\rho \| \sigma) = \tr[(\sigma^{\frac{1-\alpha}{2\alpha}} \rho \sigma^{\frac{1-\alpha}{2\alpha}})^{\alpha}].
\end{equation}
For $\rho$ and $\sigma$ both positive, these quantities were shown to be monotonic under quantum channels in the range $\alpha \in [1, 2]$ in \cite{muller2013quantum, wilde2014strong}, for all $\alpha \geq 1$ in \cite{beigi2013sandwiched}, and most generally for all $\alpha \geq 1/2$ in \cite{frank2013monotonicity}. The restriction that $\rho$ be positive, however, is needed for the monotonicity proofs presented in \cite{muller2013quantum, wilde2014strong, beigi2013sandwiched} only when $\alpha$ is not an even integer, and then only because when $\alpha$ is not an even integer, $\tilde{D}_{\alpha}$ is not necessarily well defined for non-positive $\rho$. If $\rho$ is non-positive, then the operator $\sigma^{\frac{1-\alpha}{2\alpha}} \rho \sigma^{\frac{1-\alpha}{2\alpha}}$ may be non-positive, and so for non-integer values of $\alpha$, the $\alpha$-power may be multi-valued. Furthermore, if $\alpha$ is an \emph{odd} integer, then while $\tilde{Q}_{\alpha}$ is perfectly well defined for arbitrary Hermitian $\rho$, it may be negative, and so the logarithm appearing in $\tilde{D}_{\alpha}$ can be multi-valued. When $\alpha$ is an even integer and $\rho$ an arbitrary Hermitian matrix, however, monotonicity of $\tilde{D}_{2}$ follows immediately from the proof techniques of \cite{muller2013quantum, wilde2014strong}, and monotonicity of $\tilde{D}_{2k}$ for arbitrary integer $k \geq 1$ follows from the techniques of \cite{beigi2013sandwiched}. In fact, it was observed in Lemma 2 of \cite{wang2020alpha} that the quantities
\begin{equation}
	\hat{Q}_{\alpha}(\rho \| \sigma) = \tr\left[ \left|\sigma^{\frac{1-\alpha}{2\alpha}} \rho \sigma^{\frac{1-\alpha}{2\alpha}}\right|^{\alpha} \right],
\end{equation}
with $|X| \equiv \sqrt{X^{\dagger} X},$ satisfy a monotonicity inequality like \eqref{eq:monotonicity} for arbitrary positive $\sigma$, Hermitian $\rho,$ and $\alpha \geq 1$. For non-even-integer values of $\alpha$ this function is not smooth in $\rho$ when one of the eigenvalues of $\rho$ vanishes, and so is not suitable for analysis using the methods described in this section. For even integers $\alpha$, however, and in particular for $\alpha=2$, we have $\hat{Q}_{\alpha} = \tilde{Q}_{\alpha},$ and the corresponding sandwiched R\'{e}nyi relative entropy is monotonic for arbitrary Hermitian $\rho$ and positive $\sigma$. 

\section{Approximate Petz recovery from gradient operators}
\label{sec:approximate-section}

We would like to show that when $f_{\sigma, \Lambda}$ is close to zero as a number, the gradient $\Del f_{\sigma, \Lambda}$ must be close to zero as an operator. Via equation \eqref{eq:Q2-gradient}, this would imply that $\rho$ is close to $(\mathcal{R}_{\sigma, \Lambda} \circ \Lambda)(\rho)$ as an operator, which is the statement of approximate Petz recovery.

For general functions $F$, it is not true that proximity of $F$ to a minimum implies smallness of the gradient $\Del F$. It is true, however, for functions that have \emph{bounded second derivative}. It makes sense that $\tilde{Q}_{2}(\rho \|\sigma)$ would have this property, since it is quadratic in $\rho$ and its second derivative should therefore be constant. We will make this precise momentarily. First, however, we prove a version of our desired theorem for single-variable real functions that will be essential to proving an analogous statement for the operator function $f_{\sigma, \Lambda}.$

\begin{lemma} \label{lem:bike-lemma}
	Let $f : \mathbb{R} \rightarrow \mathbb{R}$ be a twice-differentiable function satisfying $f(x) \geq y_0$ and $f''(x) \leq b$ for all real $x$. Then for all points $x_0 \in \mathbb{R}$, we have
	\begin{equation}
		|f'(x_0)| \leq \sqrt{2 b (f(x_0) - y_0)}.
	\end{equation}
\end{lemma}
\begin{proof}
	A charming way of thinking about this lemma was suggested to us by Kfir Dolev, which we present here to aid the reader's intuition: suppose Alice is riding a bike along a trail when she spots a fence some distance away from her. The variable $x$ represents time, the position of her bike is the function $f(x)$, and the position of the fence is the minimum value $y_0$ that she cannot exceed. The bound $f''(x) \leq b$ represents a maximum rate of deceleration --- the brakes on her bike can only change her speed at some constant rate. If her instantaneous velocity $f'(x_0)$ is too high at the moment she spots the fence, then even braking at full force will not be enough to bring her to a smooth stop and avoid a collision. The bound $|f'(x_0)| = \sqrt{2 b (f(x_0) - y_0)}$ represents a situation in which her initial speed is as high as possible while still allowing her to brake fully before she reaches the fence.
	
	Formally, we define the function $g(x)$ by
	\begin{equation}
		g(x) \equiv f(x_0) + f'(x_0) (x - x_0) + \frac{1}{2} b (x - x_0)^2.
	\end{equation}
	In the language of our biking analogy, this is the position curve followed by Alice if she brakes as quickly as possible. Because the zeroth and first derivative of $g(x)$ and $f(x)$ agree at $x_0$, and because the second derivative of $f(x)$ cannot exceed the second derivative of $g(x)$, we have
	\begin{align}
		g(x) - f(x)
			& = \int_{x_0}^{x} \d x' \int_{x_0}^{x'} \d x'' [g''(x'') - f''(x'')] \nonumber \\
			& = \int_{x_0}^{x} \d x' \int_{x_0}^{x'} \d x'' [b - f''(x'')] \nonumber \\
			& \geq 0.
	\end{align}
	The minimum of $g$ is achieved at $x_1 = x_0 - f'(x_0)/b$, and is given by $g(x_1) = [f(x_0) - f'(x_0)^2/2b].$ The inequality $g(x_1) \geq f(x_1)$, together with the minimum $f(x_1) \geq y_0,$ then gives the desired bound
	\begin{equation}
		|f'(x_0)| \leq \sqrt{2 b (f(x_0) - y_0)}.
	\end{equation}
	\qed
\end{proof}

To apply the underlying principle of lemma \ref{lem:bike-lemma} to the operator function $f_{\sigma, \Lambda}$, we must define what is meant by the second derivative of a function on a manifold of operators. In section \ref{sec:conceptual-section}, we defined the derivative of a function from a manifold $\mathcal{M}$ to a manifold $\mathcal{N}$ as a linear map of tangent spaces: $\d F|_{p} : T_p \mathcal{M} \rightarrow T_{F(p)} \mathcal{N}.$ The second derivative is a bilinear map on \emph{two} copies of the tangent space:\footnote{The symbol $\d^{(2)}$ appearing here is a second derivative, not the square of an exterior derivative; readers used to working with differential forms should not be tricked into thinking it satisfies $\d^{(2)} F = 0.$} $\d^{(2)} F|_p : T_p \mathcal{M} \times T_p \mathcal{M} \rightarrow T_{F(p)} \mathcal{N}.$ On manifolds with a local vector space structure, the second derivative can be defined simply as
\begin{equation} \label{eq:second-derivative-definition}
	\d^{(2)} F|_{p}(A, B)
		= \lim_{\epsilon, \eta \rightarrow 0}
			\frac{F(p+\epsilon A+\eta B) - F(p + \epsilon A) - F(p + \eta B) + F(p)}{\epsilon \eta}.
\end{equation}
If $F$ is a real-valued function, then $\d^{(2)} F$ is a rank-(0,2) tensor on the tangent space, and can be thought of equivalently as a linear map from $T_p \mathcal{M}$ to the dual space $T_p^* \mathcal{M}.$ We could proceed by defining a notion of boundedness of $\d^{(2)} F$ via the operator norm of that linear map, and proving a version of lemma \ref{lem:bike-lemma} in complete generality. For our purposes, though, this is overkill --- to bound the operator $\Del F|_{p}$, we do not need the second derivative of $F$ to be bounded as a matrix; we only need it to be bounded within the linear manifold $p + t \Del F|_{p}$ for real $t$. (Remember that since $p$ is a point in a manifold of operators, it is an operator, and so $p + t \Del F|_{p}$ is a linear manifold.)

\begin{lemma} \label{lem:operator-2-derivatives}
	Let $F$ be a twice-differentiable function from $\herm(\H)$ into $\mathbb{R}$, let $A$ be a Hermitian matrix, and define the gradient $\Del F|_{A}$ as in equation \eqref{eq:gradient-definition}. Suppose further that on the linear manifold $A + t \Del F|_{A}$, $F$ is lower bounded by $y_0$ and $\d^{(2)} F$ satisfies
	\begin{equation} \label{eq:second-derivative-gradient-bound}
		\d^{(2)} F|_{A + t \Del F|_{A}} \left( \Del F|_{A}, \Del F|_{A} \right)
			\leq b.
	\end{equation}
	Then the gradient of $F$ at $A$ satisfies
	\begin{equation}
		\tr[(\Del F|_{A})^2] \leq \sqrt{2 b (F(A) - y_0)}.
	\end{equation}
\end{lemma}
\begin{proof}
	On the line $A + t \Del F|_{A}$, the function $F$ can be thought of as a real function of $t$. The second derivative with respect to $t$ can be computed directly using ordinary single-variable calculus, and comparison with equation \eqref{eq:second-derivative-definition} gives the identity
	\begin{equation}
		\frac{\d^2}{\d t^2} F(A + t \Del F|_{A})
			= \d^{(2)} F|_{A + t \Del F|_{A}} \left( \Del F|_{A},
		\Del F|_{A} \right).
	\end{equation}
	Computing the first derivative gives
	\begin{equation}
		\frac{\d}{\d t} F(A + t \Del F)
			=  \d F|_{A + t \Del F}(\Del F|_{A}).
	\end{equation}
	So applying lemma \ref{lem:bike-lemma} at $t=0$ gives
	\begin{equation}
		\left |\d F|_{A}(\Del F|_A) \right| \leq \sqrt{2 b (F(A) - y_0)}.
	\end{equation}
	Rewriting this using equation \eqref{eq:gradient-definition} gives the desired inequality
	\begin{equation}
		\tr[(\Del F|_{A})^2] \leq \sqrt{2 b (F(A) - y_0)}.
	\end{equation}
	\qed
\end{proof}

We will now set our generic function $F$ from lemma \ref{lem:operator-2-derivatives} equal to the function $f_{\sigma, \Lambda}$ from equation \eqref{eq:f-function}. As emphasized in the final paragraph of section \ref{sec:conceptual-section}, $f_{\sigma, \Lambda}$ is lower-bounded by zero everywhere on the domain $\herm(H).$ Its second derivative can be computed directly using equation \eqref{eq:second-derivative-definition}, and is given by
\begin{equation}
	\d^{(2)} f_{\sigma, \Lambda}|_{A}(M, N)
		= 2 \tr \left( M \sigma^{-1/2} N \sigma^{-1/2} \right)
			- 2 \tr\left( \Lambda(M) \Lambda(\sigma)^{-1/2} \Lambda(N) \Lambda(\sigma)^{-1/2} \right).
\end{equation}
Note that the second derivative does \emph{not} depend on the base point $A$, which we expected from the fact that $f_{\sigma, \Lambda}$ is quadratic in $A$. In the special case $N=M$, we have
\begin{equation} \label{eq:double-M-f-intermediate}
	\d^{(2)} f_{\sigma, \Lambda}|_{A}(M, M)
		= 2 \tr \left( (\sigma^{-1/4} M \sigma^{-1/4})^2 \right)
		- 2 \tr\left[ (\Lambda(\sigma)^{-1/4} \Lambda(M) \Lambda(\sigma)^{-1/4})^2 \right],
\end{equation}
or, more simply, 
\begin{equation}
	\d^{(2)} f_{\sigma, \Lambda}|_{A}(M, M) = 2 f_{\sigma, \Lambda}(M).
\end{equation}
Because the second term of equation \eqref{eq:double-M-f-intermediate} is non-positive, we immediately obtain the upper bound
\begin{equation} \label{eq:double-M-f-penultimate}
	\d^{(2)} f_{\sigma, \Lambda}|_{A}(M, M)
		\leq 2 \left(\lVert \sigma^{-1/4} M \sigma^{-1/4} \rVert_2\right)^2.
\end{equation}
where we have introduced the \emph{two-norm} $\lVert A \rVert_{2} = \sqrt{\tr(A^{\dagger} A)}.$

The full set of Schatten $p$-norms, defined by
\begin{equation} \label{eq:schatten-norms}
	\lVert A \rVert_{p} = \tr(|A|^p)^{1/p},
\end{equation}
satisfy a family of H\"{o}lder inequalities. For $p, q, r \in [0, \infty]$ with $1/p + 1/q = 1/r$, we have
\begin{equation}
	\lVert A \rVert_r \leq \lVert A \rVert_p \lVert A \rVert_q.
\end{equation}
Applying H\"{o}lder inequalities successively to \eqref{eq:double-M-f-penultimate} yields the inequality
\begin{equation}
	\d^{(2)} f_{\sigma, \Lambda}|_{A}(M, M)
		\leq 2 (\lVert M \rVert_2)^2 (\lVert \sigma^{-1/4} \rVert_{\infty})^4.
\end{equation}
We have now proved all the inequalities needed to apply lemma \ref{lem:operator-2-derivatives} to the function $f_{\sigma, \Lambda},$ which results in the inequality
\begin{equation}
	(\lVert \Del f_{\sigma, \Lambda}|_A \rVert_2)^2
		\leq 2 \sqrt{f_{\sigma, \Lambda}(A)}  \times \lVert \Del f_{\sigma, \Lambda}|_{A} \rVert_2 \times (\lVert \sigma^{-1/4}\rVert_\infty)^2,
\end{equation}
which simplifies to
\begin{equation} \label{eq:pretty-good-inequality}
	\lVert \Del f_{\sigma, \Lambda}|_A \rVert_2
		\leq 2 \sqrt{f_{\sigma, \Lambda}(A)} (\lVert \sigma^{-1/4}\rVert_\infty)^2,
\end{equation}

To turn inequality \eqref{eq:pretty-good-inequality} into our claimed result \eqref{eq:main-bound}, we use equation \eqref{eq:Q2-gradient} for the gradient in terms of the Petz recovery channel, which gives the one-norm quality of Petz recovery as
\begin{equation}
	\lVert \rho - (\mathcal{R}_{\sigma, \Lambda} \circ \Lambda)(\rho) \rVert_1
		= \frac{1}{2} \lVert \sigma^{1/2} \Del f_{\sigma, \Lambda}|_{\rho} \sigma^{1/2} \rVert_1.
\end{equation}
Applying the H\"{o}lder inequalities to this expression gives
\begin{equation}
	\lVert \rho - (\mathcal{R}_{\sigma, \Lambda} \circ \Lambda)(\rho) \rVert_1
		\leq \frac{1}{2} (\lVert \sigma^{1/2} \rVert_4)^2 \lVert \Del f_{\sigma, \Lambda}|_{\rho} \rVert_2.
\end{equation}
We may now apply inequality \eqref{eq:pretty-good-inequality} to the last factor on the right-hand side to obtain the inequality
\begin{equation}
	\lVert \rho - (\mathcal{R}_{\sigma, \Lambda} \circ \Lambda)(\rho) \rVert_1
		\leq \sqrt{f_{\sigma, \Lambda}(\rho)} \times (\lVert \sigma^{1/2} \rVert_4)^2 \times (\lVert \sigma^{-1/4}\rVert_\infty)^2.
\end{equation}
As a final simplification, we may apply the identity $(\lVert \sigma^q \rVert_{p})^{1/q} = \lVert \sigma \rVert_{pq},$ which is straightforward to verify from the definition \eqref{eq:schatten-norms}. This gives $(\lVert \sigma^{1/2} \rVert_4)^2 = \lVert \sigma \rVert_2$ and $(\lVert \sigma^{-1/4}\rVert_\infty)^2 = \sqrt{\lVert \sigma^{-1} \rVert_\infty},$ from which we obtain the bound
\begin{equation} \label{eq:general-inequality}
	\lVert \rho - (\mathcal{R}_{\sigma, \Lambda} \circ \Lambda)(\rho) \rVert_1
		\leq \sqrt{f_{\sigma, \Lambda}(\rho)} \times \lVert \sigma \rVert_2 \times \sqrt{\lVert \sigma^{-1} \rVert_\infty}.
\end{equation}

Inequality \eqref{eq:general-inequality} is a general inequality governing the one-norm quality of the Petz map in terms of the second sandwiched R\'{e}nyi relative entropy. The function $f_{\sigma, \Lambda}(\rho)$ is what we called $\Delta \tilde{Q}_{2}$ in the introduction --- the amount that $\tilde{Q}_{2}$ changes under application of a quantum channel. When its square root is small compared to $1/(\lVert \sigma \rVert_2 \times \sqrt{\lVert \sigma^{-1} \rVert_\infty})$, inequality \eqref{eq:general-inequality} implies that the $\sigma$-Petz map does a good job inverting $\Lambda$ on $\rho$, as measured by the one norm. To make contact with inequality \eqref{eq:main-bound} from the introduction, we set the Hilbert space dimension to $d$ and let $\sigma$ be the maximally mixed state $\sigma = I/d.$ In this case, inequality \eqref{eq:general-inequality} reduces to
\begin{equation}
	\lVert \rho - (\mathcal{R}_{I/d, \Lambda} \circ \Lambda)(\rho) \rVert_1
		\leq \sqrt{\Delta \tilde{Q}_{2, I/d}}.
\end{equation}
To obtain a bound in terms of $\tilde{D}_{2}$ rather than $\tilde{Q}_{2}$, we use the equality
\begin{equation}
	\Delta \tilde{Q}_{2, I/d} = e^{\tilde{D}_{2}(\rho \| I/d)} (1 - e^{-\Delta \tilde{D}_{2, I/d}}),
\end{equation}
together with the maximal entropy bound\footnote{This follows from the fact that $\rho$ is positive semidefinite and unit-trace, which implies $\tr(\rho^2) \leq 1.$}
\begin{equation}
	\exp\left[ \tilde{D}_{2} (\rho \lVert I/d) \right]
		\leq d,
\end{equation}
which yields
\begin{equation}
	\lVert \rho - (\mathcal{R}_{I/d, \Lambda} \circ \Lambda)(\rho) \rVert_1
		\leq d^{1/2} \left(1 - e^{-\Delta \tilde{D}_{2, I/d}} \right)^{1/2}.
\end{equation}
This is exactly the bound we claimed in inequality \eqref{eq:main-bound}.

\begin{acknowledgements}
	We thank Mark Wilde for many insightful conversations about the properties of sandwiched R\'{e}nyi relative entropies, and for comments on an early version of this manuscript. We thank Kfir Dolev for providing intuition that helped improve the presentation of lemma \ref{lem:bike-lemma}.
\end{acknowledgements}

\appendix 

\section*{Appendix: Numerics}
\label{sec:numerics}

It would be interesting to understand whether inequalities \eqref{eq:main-bound} and \eqref{eq:Q-bound} are tight. That is, for fixed Hilbert space dimension $d$ and arbitrary allowed values of $\Delta \tilde{D}_{2,I/d}$ and $\Delta \tilde{Q}_{2, I/d}$, we would like to know if there exists a state $\rho$ and channel $\Lambda$ satisfying
\begin{equation}
	\lVert \rho - (\mathcal{R}_{I/d, \Lambda} \circ \Lambda)(\rho) \rVert_1
		= d^{1/2} \left(1 - e^{-\Delta \tilde{D}_{2, I/d}} \right)^{1/2}
\end{equation}
or 
\begin{equation}
	\lVert \rho - (\mathcal{R}_{I/d, \Lambda} \circ \Lambda)(\rho) \rVert_1
		= \sqrt{\Delta \tilde{Q}_{2, I/d}}.
\end{equation}
A detailed understanding of this question is beyond the scope of the present work. We present only a simple numerical test of the question in dimensions $d=2$ and $d=3.$ We generate states $\rho$ and channels $\Lambda$ at random, compute the left- and right-hand sides of \eqref{eq:Q-bound} and \eqref{eq:main-bound}, and compare them.

To generate a random state $\rho$, we start with a complex matrix $M$ whose entries have real and imaginary parts drawn independently from the unit-variance Gaussian distribution, and set $\rho = M^{\dagger} M / \tr(M^{\dagger} M).$ To generate random channels $\Lambda$, we make use of Kraus's theorem \cite{kraus1971general} that $\Lambda$ is a quantum channel if and only if it can be written as
\begin{equation}
	\Lambda(A)
		= \sum_j E_j A E_j^{\dagger}
\end{equation}
for some set of operators $\{E_j\}$ that has at most $d^2$ elements and satisfies $\sum_j E_j^{\dagger} E_j = I.$ Following a prescription from \cite{kukulski2021generating}, we pick an integer $n$ from $1$ to $d^2$ randomly, generate $n$ complex matrices $F_j$ by drawing the real and imaginary parts of the matrix entries independently from the unit-variance Gaussian distribution, and define
\begin{equation}
	E_j \equiv F_j \left( \sum_k F_k^{\dagger} F_k \right)^{-1/2}
\end{equation}
to guarantee $\sum_j E_j^{\dagger} E_j = I.$

For each of the cases $d=2$ and $d=3$, we generated 100,000 random states $\rho$ and channels $\Lambda$. The left-hand side of figure \ref{fig:2-plots} shows a scatter plot of $\lVert \rho - (\mathcal{R}_{I/d, \Lambda} \circ \Lambda)(\rho) \rVert_1$ against $\Delta \tilde{D}_{2,I/d}$ in $d=2$, imposed over the region allowed by our inequality \eqref{eq:main-bound}. The right-hand side of the same figure shows a scatter plot of $\lVert \rho - (\mathcal{R}_{I/d, \Lambda} \circ \Lambda)(\rho) \rVert_1$ against $\Delta \tilde{Q}_{2, I/d}$, imposed over the region allowed by our inequality \eqref{eq:Q-bound}. Figure \ref{fig:3-plots} shows analogous plots in $d=3.$ For $d=2$, both bounds appear quite tight, in the sense that there are samples all along the upper edge of the allowed region. For $d=3$, neither bound seems perfectly tight. Whether this is genuinely the case, or whether the failure of our numerical samples to saturate the bounds is the result of some concentration of measure effect on the particular ensemble we have sampled, would have to be established in future work using either analytical methods or a more careful numerical search.

\begin{figure}
	\centering
	\includegraphics[scale=0.4]{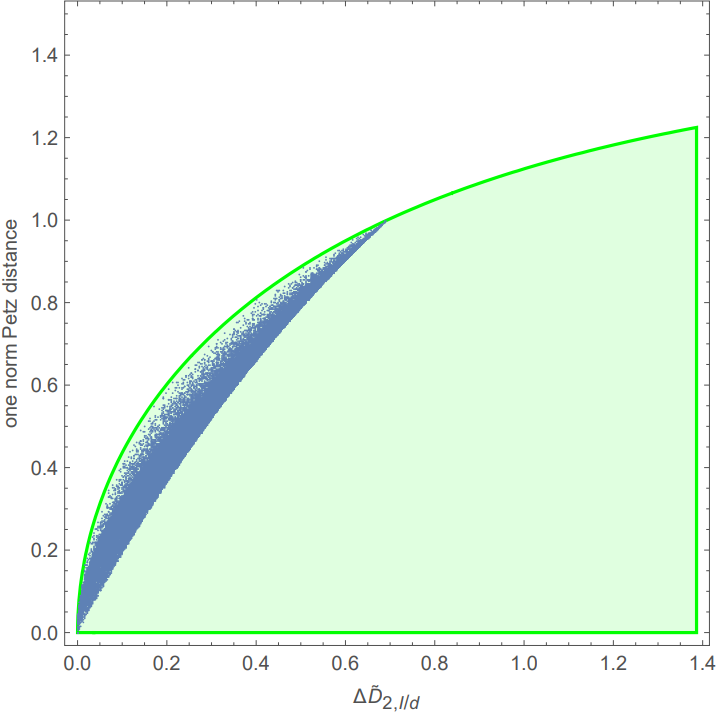}
	\hspace{0.5cm}
	\includegraphics[scale=0.4]{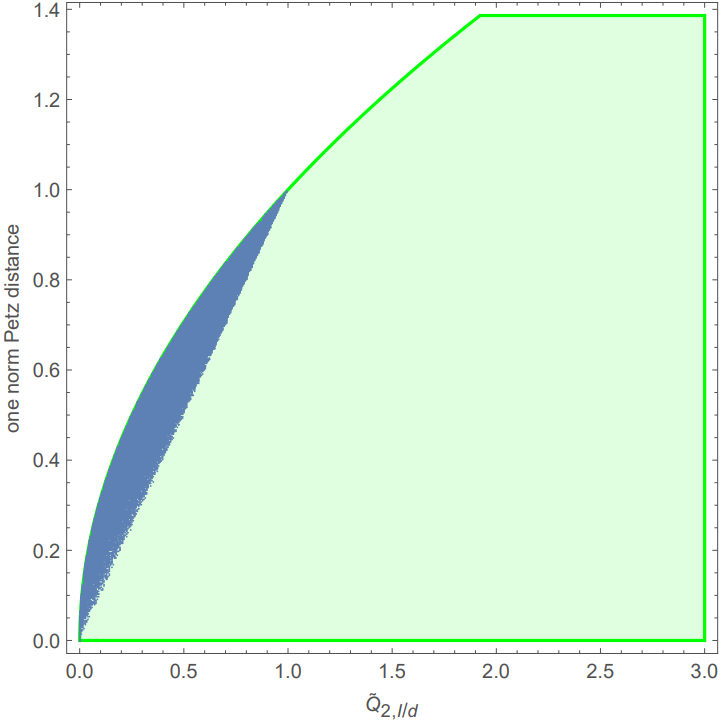}
	\caption{Left: The region in the $(\lVert\rho - (\mathcal{R}_{I/d, \Lambda} \circ \Lambda)(\rho)\rVert_1, \Delta \tilde{D}_{2, I/d})$ plane allowed by inequality \eqref{eq:main-bound} for $d=2$, together with a scatter plot of those quantities for 100,000 random states $\rho$ and channels $\Lambda$. Right: An analogous plot but in terms of $\tilde{Q}_{2,I/d}$ and inequality \eqref{eq:Q-bound}.}
	\label{fig:2-plots}
\end{figure}

\begin{figure}
	\centering
	\includegraphics[scale=0.4]{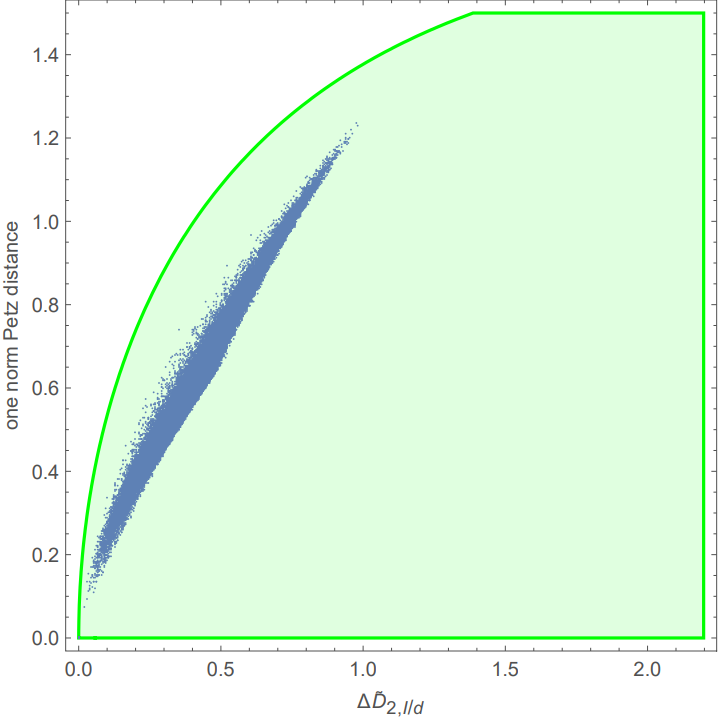}
	\hspace{0.5cm}
	\includegraphics[scale=0.4]{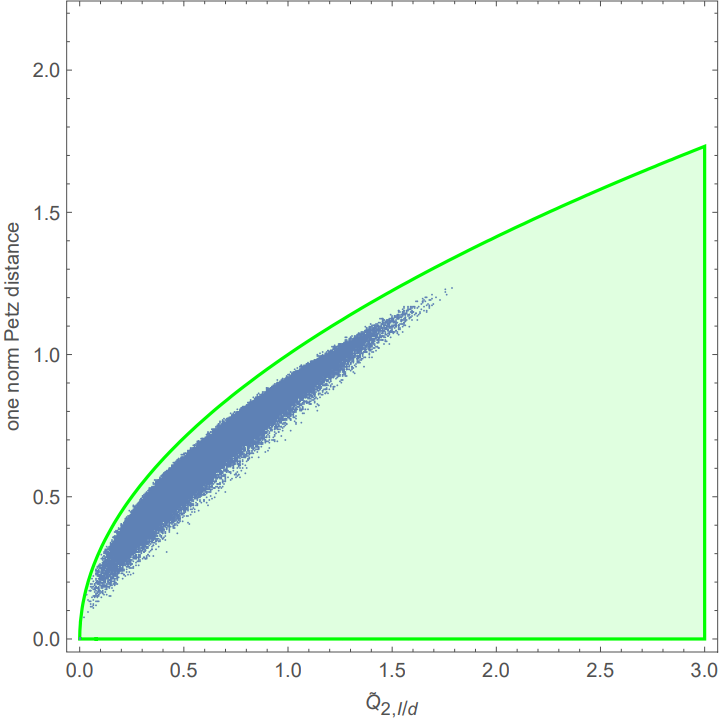}
	\caption{Left: The region in the $(\lVert\rho - (\mathcal{R}_{I/d, \Lambda} \circ \Lambda)(\rho)\rVert_1, \Delta \tilde{D}_{2, I/d})$ plane allowed by inequality \eqref{eq:main-bound} for $d=3$, together with a scatter plot of those quantities for 100,000 random states $\rho$ and channels $\Lambda$. Right: An analogous plot but in terms of $\tilde{Q}_{2,I/d}$ and inequality \eqref{eq:Q-bound}.}
	\label{fig:3-plots}
\end{figure}

\bibliographystyle{spphys}       
\bibliography{biblio.bib}   

\end{document}